\documentclass[11pt,a4paper]{article}
\usepackage[utf8]{inputenc} 
\usepackage[english]{babel}
\usepackage[T1]{fontenc}
\usepackage{caption}
\usepackage{graphicx}
\usepackage{bbm}

\usepackage{colortbl}

\usepackage{algorithm}
\usepackage{algorithmic}

\usepackage[colorlinks,linkcolor=blue,filecolor=blue,citecolor=blue,urlcolor=blue,pdfstartview=FitH,hypertexnames=false]{hyperref}
\usepackage{amsmath,amsthm,amssymb}

\newcommand*\diff{\mathop{}\!\mathrm{d}}

\usepackage{geometry}

\usepackage{tikz,fullpage}
\usetikzlibrary{arrows.meta}
\usetikzlibrary{matrix}
\usepackage{pgfplots}
\usetikzlibrary{positioning}
\definecolor {processblue}{cmyk}{0.96,0,0,0}

\newtheorem{mydef}{Definition}
\newtheorem{theorem}{Theorem}
\newtheorem{lemma}{Lemma}
\newtheorem{problem}{Problem}
\newtheorem{prop}{Property}
\newtheorem{rem}{Remark}
\graphicspath{{C:/Users/pierre/OneDrive/Documents/These/Figures/}}

\makeatletter
\newcommand{\xMapsto}[2][]{\ext@arrow 0599{\Mapstofill@}{#1}{#2}}
\def\Mapstofill@{\arrowfill@{\Mapstochar\Relbar}\Relbar\Rightarrow}
\makeatother

\usepackage{authblk}
 
\title{Continuous \textit{indetermination} and average likelihood minimization}
\date{\today}
\author[,1]{Pierre Bertrand\thanks{Corresponding author: pierre.bertrand@ens-cachan.fr}}
\author[1,2]{Michel Broniatowski}
\author[3]{Jean-François Marcotorchino}
 
\affil[1]{Laboratoire de Probabilités, Statistique et Modélisation, Sorbonne Université, Paris, France}
\affil[2]{CNRS UMR 8001, Sorbonne Université, Paris, France}
\affil[3]{Institut de statistique, Sorbonne Université, Paris, France}
\begin{document}
\maketitle

\begin{abstract}
The authors transpose a discrete notion of \textit{indetermination} coupling in the case of continuous probabilities. They show that this coupling, expressed on densities, cannot be captured by a specific copula which acts on cumulative distribution functions without a high dependence on the margins. Furthermore, they define a notion of average likelihood which extends the discrete notion of couple matchings and demonstrate it is minimal under indetermination. Eventually, they leverage this property to build up a statistical test to distinguish indetermination and estimate its efficiency using the Bahadur's slope.
\end{abstract}


\section{Introduction}

A key tool in statistics is the likelihood estimator which consists, given a realization $W=w$ of a random variable $W$ to estimate the probability under a candidate law $P$ to have $W'\sim P=w$. Typically, if $P_{\pi}$ has a density function $\pi$ it is represented by $\pi(w)\diff w$. In a discrete case it would represent the probability to have an exact match between $W$ and $W'$. Precisely, in a precedent paper~\cite{bertrand:hal-03086553} we demonstrated how \textit{indetermination} concept could reduce those matchings in the discrete case. This paper is interested in the continuous transposition of this notion. Although the definition of the continuous indetermination structure is quite simple at first glance, it is motivated by an anticipated link with the likelihood reduction as the continuous extension of discrete couple matchings and conveys, actually, some interesting properties. 

We begin with the introduction of \textit{indetermination}. Given two discrete marginal laws $\mu = \mu_1\ldots\mu_p$ and $\nu=\nu_1\ldots\nu_q$, a coupling function $C$ constructs a probability law $\pi$ on the product space where $\pi_{u,v}=C(\mu_u,\nu_v)$. Theoretical considerations based on a work of Csisz\'ar \cite{csiszar1991least}, a summarized version of which is expressed in \cite{bertrand2020statistical} lead to consider two "natural" equilibria. Shortly, Csisz\'ar works on projection functions: given a set of constraints, a first guess has to be projected into the eligible space. Then, adding some natural properties to be respected by the projection function, it leads to consider only two possible criteria : either least square or maximum entropy. Eventually, under some reformulations each projection function is associated to a specific coupling function, respectively: the usual independence $\pi^{\times}_{u,v} = C^{\times}(\mu,\nu)_{u,v} = \mu_u\nu_v,~\forall (u,v)$ and the so-called indetermination (first quoted as such in~\cite{MAR84}) whose formula is given by:
\begin{equation}
\pi^+_{u,v} = C^+(\mu,\nu)_{u,v} = \frac{\mu_u}{q}+\frac{\nu_v}{p}-\frac{1}{pq},~\forall(u,v)
\end{equation}

Additional constructive properties of discrete indetermination together with interpretations and applications can be found in \cite{bertrand:hal-03086553} where we develop its analysis. Some demonstrations with useful interpretations (typically justifying the name "indetermination" or "indeterminacy") require the introduction of "Mathematical Relational Analysis" which is detailed in \cite{MAM79}, \cite{MAR86} or \cite{AHP09a}. We will not present it again in this paper but we refer the interested reader to the aforementioned articles. 

By construction (related to the least square cost it minimizes), $\pi^+$ is the optimal law to reduce matching in case margins are fixed. If we draw two couples $W_1=(U_1,V_1)$ and $W_2=(U_2,V_2)$ under $\pi^+$ the probability to have $W_1=W_2$ is minimal. The natural transposition in the continuous case of this matching property would be possibly the motivation for a continuous indetermination. 

Actually, the idea of a continuous indetermination appears first in~\cite{MarcoGSIDual} in terms of coupling, and in~\cite{HUY15} in terms of "indetermination copula". Let us precise that the paper~\cite{HUY15} though it conveys an interesting method to detect anomalies, lacks the theoretical foundations  of some concepts  it introduces. Paper~\cite{MarcoGSIDual} is more theoretical but it tackles the subject in a too short way. In consequence the continuous "indetermination copula" must be  properly introduced and refers to a notion that deserves to be investigated in depth, it is precisely one goal of this article.

In this paper we show that we cannot define an "indetermination copula" \textit{per se} but that it can be locally defined. Shortly, whilst an independence copula is easily defined without reference to the margins, an indetermination coupling uses a margin-dependent copula. A couple of margins to share an indetermination copula requires them to be close in a sense that we will define later in the paper. 

Additionally, we interpret the continuous indetermination as a lower bound on the average likelihood $\overline{L}$, a notion that we introduce as an equivalent of discrete couple matchings. We show that based on a list of $n$ realizations of $W$ a random variable under indetermination, the correlation vanishes under the average likelihood estimator which becomes sensitive only to both margins. Namely, if $P_n$ is the empiric measure associated with the $n$ realizations, $\overline{L}(P_n,\pi)$ only depends on the margins of $\pi$ whatever the underlying copula is. 

Finally, we define a statistical test $t_n$ to reject or accept the hypothesis of continuous indetermination based on its property to minimize the average likelihood. Following the usual track, detailed in~\cite{groeneboom1977bahadur} and dedicated to the analysis of statistical test, we compute the Bahadur's slope of $t_n$ as an estimation of its quality.

This paper is composed of two sections, section~\ref{sec:copula} introduces the continuous indetermination together with some properties and studies the notion of indetermination copula. Section~\ref{sec:likelihood} defines a notion of average likelihood $\overline{L}$ between two probabilities and demonstrates that, applied to probabilities in a product space, its value between $\pi$ against $\pi^+$ only depends on the four associated margins. It notably implies that any $\pi$ with the same margins as $\pi^+$ has the same average likelihood as $\pi^+$ against itself. Eventually, in subsection~\ref{ssec:test} we leverage on this property to build up a test to distinguish $\pi^+$ and analyse it at the light of the Bahadur's slope.

\section{Continuous indetermination \label{sec:copula}}

\subsection{Coupling function}
This chapter introduces a continuous indetermination coupling that, to our knowledge, has never been studied in depth 'see previous remarks). Extrapolating the discrete margins situation, we now suppose $\mu$ is a probability law on a segment $[a,A]$ and $\nu$ is a same type of probability but on a segment $[b,B]$. Let us also suppose they both have a density function, respectively $f$ and $g$. To define a probability measure $\pi$ we can operate on the density functions as we basically did in the discrete case (where we essentially added $\mu_u$ and $\nu_v$).

We follow and adapt the notations introduced in \cite{csiszar1984sanov} and \cite{groeneboom1979large} as well as some lemmas or theorems they cover and that we shall use in the sequel (section~\ref{sec:likelihood}). First, we begin by quoting $S_a=[a,A]$, $S_b=[b,B]$ and $S=[a,A]\times[b,B]$ together with the usual borelian set of probability measures on the three: $\Lambda_a$ set of probability measures on $(S_a,\mathcal{B})$, $\Lambda_b$ on $(S_b,\mathcal{B})$ and $\Lambda$ on $(S,\mathcal{B})$.

As a first step, we extend the definition of a coupling function in a continuous space. 

\begin{mydef}[Coupling functions (continuous case)]
\label{def:continuouscoupling}
~\\
$\mu\in\Lambda_a$ and $\nu\in\Lambda_b$ being probability laws, a coupling function $C$ operates on their density $(f,g)$ to define a density $C(\mu,\nu)=C(f,g)$ for a measure on the product space $S$. That measure respects some properties similar to the ones found in the discrete case: it is a probability law whose margins are $\mu$ and $\nu$. 
\end{mydef}

A typical coupling function is the independance quoted $C^{\times}$ (we extend here the discrete notation) which generates an eligible density under the formula:
$$C^{\times}(f,g)(x,y) = f(x)g(y),~\forall x\in S_a,~\forall y \in S_b$$

We would like to define a continuous version of the indetermination coupling. As usual, when we transpose a concept, we can either obtain it through computations or through a prior guess. We will follow both approaches. Let us first propose a prior guess. 

\begin{mydef}[Continuous indetermination density]
\label{def:continuous+}
$$C^{+}(f,g)(x,y) = \frac{f(x)}{B-b} + \frac{g(y)}{A-a} - \frac{1}{(A-a)(B-b)}$$
\end{mydef}

The formula comes from an adaptation of the discrete one, no guarantee is given on its truthfulness neither on its construction. 
First, to simplify any future computation, we set $a=b=0$ and $A=B=1$, converting any formula will be done using a dedicated affine transformation. 

We shall use an optimal transport problem to validate our prior guess, we transpose actually the discrete "Minimal Transport problem" using least square cost function into the continuous space:
\begin{problem}[Minimal Trade Problem]\label{pb:continous}
\begin{eqnarray*}
	 \min_{\pi}&& \int_0^1\int_0^1 \pi^2(x,y)\diff x \diff y\\
	 \textnormal{under constraints:}&&\\
 	 &&\int_0^1 \pi(x,y)\diff y \diff x = \mu_{x} \textnormal{ (first margin)}\\
	 &&\int_0^1 \pi(x,y)\diff x \diff y = \nu_{y} \textnormal{ (second margin)}\\
	 &&\int_0^1 \int_0^1 \pi(x,y)\diff x \diff y = 1\textnormal{ (mass preserving)}\\
	 &&\pi\geq 0
\end{eqnarray*}
\end{problem}
~\\
We then add inequality~\ref{cond:continuous} similar to the one used in the discrete case (see~\cite{bertrand:hal-03086553}) which ensures that our prior guess function is a probability law. As well as restricting ourselves to $[0,1]$ we force the margins to respect the condition:	
\begin{equation}\label{cond:continuous}
	\min_{x\in S_1} f(x) + \min_{y\in S_1} g(y) \ge 1
\end{equation}

\begin{prop}
\label{prop:minL2}
~\\
Under the \textit{ad hoc} hypothesis~\ref{cond:continuous}, the solution of problem~\ref{pb:continous} is nothing but our prior guess continuous coupling function applied to the margins $f$ and $g$, formally, the associated density function is:
$C^{+}(f,g)(x,y) = c^+_{f,g}(x,y) = (f(x)+g(y)-1)$
\end{prop}
\begin{proof}
~\\
We use definition~\ref{def:continuous+} to efficiently solve problem~\ref{pb:continous} by noticing that 
$$\int_{x=0}^1\int_{y=0}^1\left[\pi(x,y)-(f(x)+g(y)-1)\right]^2\diff x \diff y \ge 0$$
Which can be rewritten (using constraints on margins):
\begin{eqnarray*}
\int_{x=0}^1\int_{y=0}^1\pi^2(x,y)\diff x \diff y &\ge & \int_{x=0}^1\int_{y=0}^1 \left(-f^2 -g^2 -1 -2\pi -2fg +2f\pi +2g\pi +2f +2g\right) \diff x \diff y\\
&=& \int_{x=0}^1\int_{y=0}^1 \left(-f^2 -g^2 -1 -2 -2fg +2f^2 +2g^2 +2 +2\right) \diff x \diff y\\
&=& \int_{x=0}^1\int_{y=0}^1 \left(f^2 + g^2 -2fg +1\right) \diff x \diff y\\
&=& \int_{x=0}^1\int_{y=0}^1 \left(f + g -1 \right)^2 \diff x \diff y
\end{eqnarray*}
At this stage we have shown 
\begin{equation*}
\min_{\pi} \int_0^1\int_0^1 \pi^2(x,y)\diff x \diff y \geq \int_{x=0}^1\int_{y=0}^1 \left(C^+(f,g)(x,y)\right)^2 \diff x \diff y
\end{equation*}
Eventually, using hypothesis~\ref{cond:continuous} we notice that $c^+_{f,g}(x,y) = (f(x)+g(y)-1)$ is eligible as a density function since always positive. Margins constraints are also satisfied using the same considerations as in the discrete case. This provides the second inequality:
\begin{equation*}
\min_{\pi} \int_0^1\int_0^1 \pi^2(x,y)\diff x \diff y \le \int_{x=0}^1\int_{y=0}^1 \left(C^+(f,g)(x,y)\right)^2 \diff x \diff y
\end{equation*}
which allows us to conclude.
\end{proof}

\begin{mydef}[Indetermination coupling]
~\\
Given two probability laws on $S_1$($S_a$ with $a=1$): $U\sim\mu,~(f,F)$ and $V~\sim\nu,~(g,G)$ we say that the random variable $W$ is an indetermination coupling of $U$ and $V$, quoted $U\oplus V$ when its density $\pi^+$ is 
$\pi^{+}_{f,g}(x,y) = f(x)+g(y)-1$
\end{mydef}

Let us compute the cumulative distribution function $\Pi^+_{F,G}$ associated to the just expressed density $\pi^+_{f,g}$; it is only a quick integration of the density leading to a second characterization of an indetermination coupling. 

\begin{prop}[Cumulative distribution function for indetermination ]
~\\
If $\pi^+_{f,g}(x,y) = f(x)+g(y)-1$ is the density of a random variable, with $f,g$ two densities on $S_1$ ensuring $\pi^+_{f,g}$ is positive (hypothesis~\ref{cond:continuous}) and whose cumulative distribution functions are $F$ and $G$ respectively, then the associated cumulative distribution function, quoted $\Pi^+_{F,G}$, is given by:
$$\Pi^+_{F,G}= yF(x)+xG(y)-xy$$
\end{prop}

\subsection{Constructive method for adapted margins}
Beforehand, we assumed margins are respecting hypothesis~\ref{cond:continuous}, we propose here a method to construct adapted margins out of any couple. 

\begin{prop}[Constructive margins]
\label{prop:constructiveAlpha}
~\\
A couple $(f,g)$ of densities fulfills hypothesis~\ref{cond:continuous} if and only if, it exists an $\alpha, 0 \le \alpha \le 1$ and a couple of densities $(r,s)$ such that :
\begin{eqnarray}
	f = (1-\alpha)r + \alpha & \mbox{ and } & g = \alpha s + 1-\alpha
\end{eqnarray}
\end{prop}
\begin{proof}
~\\
Let us first suppose $(f,g)$ are under this form, then we notice, $\min f \geq \alpha$ as well as $\min g \geq 1-\alpha$, hence, $\min f + \min g \geq 1$ so that the condition is respected. 

Now, if the condition is respected, we define $\alpha = \min f$ and have $g\geq 1-\alpha$. If $\alpha = 0$ or $\alpha = 1$ then, respectively, $g$ or $f$ is uniform so that the corresponding coupling is independence and condition is degenerated. If not, $0< \alpha < 1$ and we can write: 
\begin{eqnarray*}
	f &=& (1-\alpha)\frac{f-\alpha}{1-\alpha} + \alpha \\
	\mbox{together with:}&&\\
	g &=& \alpha\frac{g-(1-\alpha)}{\alpha} + (1-\alpha)
\end{eqnarray*}
Quoting $r = \frac{f-\alpha}{1-\alpha}$ and $s = \frac{g-(1-\alpha)}{\alpha}$, we have $0\le r,s\le 1$ and $\int r = \int s = 1$. It precisely shows that $(r,s)$ is a couple of densities on $S_1$.
\end{proof}

Using the last proposition, we can easily build up a couple of margins on which an indetermination coupling is feasible. As mentioned during the proof, if $\alpha\in\{0,1\}$ then $f$ or $g$ is an uniform density for which indetermination and independence are completely equivalent. Hence we will exclude $\alpha\in\{0,1\}$ from most of our computations.

We repeatedly notice in the discrete case that properties of independence usually have a symmetric transcript for indetermination. In the continuous case, we know we can define an independence copula, operating on the cumulative distribution functions of the margins to generate the cumulative distribution function of an independence coupling. That property is not as common as one would expect, notably, an indetermination copula is out of existence. Let us properly explain where the difference comes from. 

\subsection{A specific indetermination copula}

We first remind the definition and properties of copulas as a classic way to couple two margins together with the well-known Sklar theorem (see~\cite{Sklar73}). 

\begin{mydef}[Copula]
~\\
\label{def:copula}
A copula $\mathbf{C}$ (in bold to distinguish from a coupling function) is a cumulative distribution function on $[0,1]^d$ ($d\in\mathbb{N}$) whose margins are uniforms. It is defined by three properties valid for any $u = (u_1,\ldots,u_d)$:
\begin{itemize}
\item $\mathbf{C}(u) = 0$ as soon as any $u_i$ is null
\item $\mathbf{C}(u) = u_i$ if $u_i$ is the only component different from $1$
\item $\mathbf{C}$ is $d-$non decreasing 
\end{itemize}
\end{mydef}

A copula is used to construct a coupling law by defining it with its cumulative distributive function while a coupling function (definition~\ref{def:continuouscoupling}) operates on densities; the close notations insist on their similarity while the type of function (density or cumulative distribution function) distinguish them. One could expect any coupling function to generate a corresponding copula. The transposition is not that easy: as we shall see it may depend on margins. 

The Sklar's theorem extracts and applies copulas to any probability law. More precisely, it indicates that any function $C$ satisfying the properties stated in Definition~\ref{def:copula} can be applied to any set of $d$ univariate cumulative functions $(F_1,\ldots,F_d)$ to generate a multivariate cumulative function whose margins will precisely be the $F_i$ respecting the formula:
$$F(x_1,\ldots,x_d) = \mathbf{C}(F_1(x_1),\ldots,F_d(x_d))$$
Reciprocally, any cumulative distribution function $F$ corresponds to an associated copula $\mathbf{C_F}$. Typically, when the margins $F_i$ of $F$ have a closed formula, we have
$$\mathbf{C_F}(u_1,\ldots,u_d) = F(F_1^{-1}(u_1),\ldots,F_d^{-1}(u_d))$$

\begin{rem}[$d-$ increasing]
~\\
The $d-$increasing property is closely tied to the positivity of the underlying probability law $\mathbb{P}$.
Hence, in dimension $2$, $\mathbb{P}(u_1\le U \le u_1', u_2\le U \le u_2') = \mathbf{C}(u_1,u_2) - \mathbf{C}(u_1,u_2') - \mathbf{C}(u1',u_2) + \mathbf{C}(u_1',u_2') \ge 0$. 
\end{rem}

\subsubsection{Extraction of the indetermination copula}
We begin by applying the Sklar's theorem. We do know that when $\mu$ and $\nu$ are fixed and fulfill hypothesis~\ref{cond:continuous}, we can couple them under indetermination operating on their densities; we keep the usual notations previously introduced for their densities as well as for their cumulative distribution functions.

With indetermination, we obtain a probability on the couple space whose density is given by: $\pi^+_{f,g} = h(x,y) = f(x)+g(y)-1,~\forall(x,y) \in S$ (remember we restricted ourselves to $S_1=[0,1]$).

Using the associated cumulative distribution function $\Pi^+_{F,G}$, we can express the associated copula. It is specific to the indetermination coupling (the way we generated $\pi^+_{f,g}$ hence $\Pi^+_{F,G}$) as well as to the margins $F$ and $G$ a priori. 
\begin{mydef}[Specific indetermination copula]\label{def:CopInd}
~\\
Given two margins $U\sim\mu,~(f,F)$ and $V\sim\nu,~(g,G)$, their specific indetermination copula, extracted from their indetermination coupling is given by:
\begin{equation}\label{eq:CopInd}
\mathbf{C_{F,G}^+}(u,v) = v*F^{-1}(u) + u*G^{-1}(v) - F^{-1}(u)*G^{-1}(v), ~\forall (u,v) \in S
\end{equation}
\end{mydef}

\begin{rem}[Extension]
~\\
$\mathbf{C_{F,G}^+}$ is defined by the margins, we shall insist on that later on. Yet, they are characterized by their densities, their cumulative distribution function as well as by their probability measure. Provided we can pass from one to another, we shall abusively speak of an indetermination copula for a couple of densities $(f,g)$ ($\mathbf{C^+_{f,g}}$) or a couple of probability measures $(\mu,\nu)$ ($\mathbf{C^+_{\mu,\nu}}$).
\end{rem}

\subsubsection{Example}
Let us introduce an example, $U$ and $V$ respectively follow the cumulative distribution function $F(x)= x^{\alpha}$ and $G(y) = y^{\beta}$. We first need to ensure that hypothesis~\ref{cond:continuous} is respected. Here it requires that the densities verify:
\begin{eqnarray*}
\min_x f + \min_y g -1 &\ge& 0\\
\min_x{\alpha*x^{\alpha-1}} + \min_y{\beta y^{\beta-1}} &\ge& 1\\
\textnormal{hypothesis~\ref{cond:continuous} leads to $\alpha$ and $\beta$ less than $1$ and }&&\\
\alpha + \beta &\ge& 1
\end{eqnarray*}
Under that last hypothesis, 
\begin{equation}\label{cop_ab}
\mathbf{C_{\alpha,\beta}^+}(u,v) = uv^{\frac{1}{\beta}} + u^{\frac{1}{\alpha}}v -u^{\frac{1}{\alpha}}v^{\frac{1}{\beta}}
\end{equation}

\begin{rem}[Dependence on margins]
~\\
A remark that we shall develop later: the presence of $\alpha$ (as well of $\beta$) inside the formula prevents it from being independent from the margins.
\end{rem}

So far, we did not define a copula of indetermination \textit{per se} but a specific copula of indetermination among those satisfying the formula of Definition~\ref{def:CopInd} for a given couple of margins.

\subsubsection{Dependence on margins}
We unfold the link between a specific indetermination copula and its margins. The copula is a way among many to couple two laws (maybe more than two but we limit ourselves). We already defined two of them through Definition~\ref{def:continuouscoupling} which operates on densities and by Definition~\ref{def:copula} which operates on cumulative distribution function. Replacing ourselves in the context of indetermination, it leads to formula~\ref{def:continuous+}, or to apply a specific indetermination copula given by formula~\ref{eq:CopInd} to a pair of cumulative distribution function. 

The second way for coupling exhibits a problem that we can easily isolate; it depends on margins as functions $F^{-1}$ and $G^{-1}$ appear in the formula used to mimic the application on densities to cumulative distribution function. We notice that the copula associated to an indetermination coupling depends on the margins we want to couple. It is not the case (for instance) when using an independence coupling. 

In general, given two margins, one can apply at least three coupling ways: 
\begin{enumerate}
\item define a density $\pi(x,y)$ verifying the margins (typically applying a coupling function $C$)
\item define a cumulative distribution function $\Pi(x,y)$ verifying the margins
\item define a copula $\mathbf{C}$ and apply it to the margins (verified by construction)
\end{enumerate}

In case of independence they correspond to $C^{\times}(u,v) = uv$, $\Pi^{\times}(x,y) = F(x)G(y)$ and $\mathbf{C}^{\times}(u,v) = uv$. It is a particular concept where any of the three coupling methods shall be defined without reference to the associated margins.

For indetermination, it doesn't work as $C(u,v) = u+v-1$, $\Pi^+(x,y) = xG(y)+yF(x)-xy$ but $\mathbf{C^+}(u,v) = u*F^{-1}(v) + v*G^{-1}(u) - F^{-1}(u)*G^{-1}(u)$ is highly dependent on margins. It is quite visible on the example~\ref{cop_ab} which takes various forms when $\alpha$ and $\beta$ vary. The only constraints of those two parameters being less than $1$ and with a sum greater than $1$.

Going back to independence (defined on density), we estimate why it is so specific. $\mathbf{C}$ being a copula, we suppose $u=F(x)$ and  $v=G(y)$ for more readiness. By definition, if we quote $\Pi$ the cumulative distribution function and $\pi$ the density:
\begin{align*}
\mathbf{C}(F(x),G(y)) &= \Pi(x,y)\\
\intertext{differentiating with respect to $x$ and $y$, we get:}\\
\frac{\partial \mathbf{C}(u,v)}{\partial x \partial y} &= \frac{\partial \Pi}{\partial x \partial y}\\
\frac{\partial u}{\partial x}\frac{\partial v}{\partial y}\frac{\partial C(u,v)}{\partial u \partial v} &= \pi(x,y)\\
f(x)g(y)\frac{\partial \mathbf{C}(u,v)}{\partial u \partial v} &= \pi(x,y)\\
\frac{\partial \mathbf{C}(u,v)}{\partial u \partial v} &= \frac{\pi(x,y)}{f(x)g(y)}
\end{align*}

Hence, independence coupling is quite peculiar as margins vanish: $\frac{fg}{fg}=1$ is the second derivative of the associated copula $\mathbf{C^{\times}}$ which is precisely independent of $F$ and $G$.

For an indetermination coupling, the crossed derivative of the copula $\mathbf{C^{+}}$ is: $\frac{f+g-1}{fg} = \frac{1}{g} + \frac{1}{f} -\frac{1}{fg}$ whose integration leads to formula~\ref{eq:CopInd} using $\frac{\partial F^{-1}(u)}{\partial u} = \frac{1}{F'(F^{-1}(u))} = \frac{1}{f(x)}$. It is dependent on margins and one cannot define a generic indetermination copula. We shall, so far, only define a specific indetermination copula $\mathbf{C^+_{F,G}}$

\begin{rem}
~\\
These last computations require  a division by $fg$. Let us go back to hypothesis~\ref{cond:continuous}. We know that $f+g-1\geq 0$ so that if there exists $x$ such that $f(x)=0$ then $\forall y,~g(y) \geq 1$ and, as $\int_{0}^1 g = 1$ (we are on segment $[0,1]$, if not we have to divide by the length as in Definition~\ref{def:continuous+}), we deduce $g=1$. It leads to a poorly interesting coupling: $f+g-1=f$. Moreover, the computations are actually still valid since $\mathbf{C_{F,G}^+} = yF^{-1}(x)$ whose crossed derivative is nothing but: $\frac{1}{f}$.
\end{rem}

\subsection{A local indetermination copula}
In this section we measure the dependence of $\mathbf{C^+_{F,G}}$ on its margins. The idea is to check whether a given specific copula of indetermination can create indetermination coupling when applied to another couple of margins than the one which defines it. The property~\ref{prop:copuleLocale} shows that it is locally possible.

To answer, let us start with four densities $f,g,r,s$ as well as their respective cumulative distribution functions $F,G,R,S$.

First step, we couple $(f,g)$ through indetermination using the indetermination coupling function of Definition~\ref{def:continuous+} and extract the associated specific copula of indetermination $\mathbf{C_{F,G}^+}$ using Definition~\ref{def:CopInd}. As already explained, it can be applied to any couple of margins, hence the second step.

Second step, we apply $\mathbf{C_{F,G}^+}$ to the margins defined by $(R,S)$, it leads to:
\begin{equation}
\label{eq:composed_copula}
\Pi^{f,g}_{r,s} = \mathbf{C^+_{F,G}}(R,S) = H
\end{equation}
We shall quote $H$ this cumulative distribution function throughout this section, it can be applied to any $(x,y)$ in $S$.

Third and last step, we compare $H$ to the coupling of $(r,s)$ under indetermination, namely $\mathbf{C^+_{R,S}}(R,S)$.

\bigbreak

Since we do not expect a general equality it implies at least some requirements on $(F,G,R,S)$ so that $H$ represents a coupling of indetermination. Obviously, if $(F,G) = (R,S)$, the third step is trivial as both functions are equal. The natural question being: are there any other possibilities? The answer is yes, provided that the requirements expressed below are satisfied.

\begin{prop}[A local indetermination copula]\label{prop:copuleLocale}~\\
Given $f,g,r,s$, four densities and $(R,G,R,S)$ the application of $\mathbf{C^+_{F,G}}$ to $(R,S)$ quoted $H$ (see equation~\ref{eq:composed_copula}) generates an indetermination coupling if and only if
it exists a $\lambda\geq 0$ respecting:
\begin{equation}\label{eq:lambda}
\frac{max(g)}{max(g)-1}\geq \lambda \geq 1 - \frac{1}{max(f)}
\end{equation}

such that the following equations are satisfied:
\begin{eqnarray}\label{eq:linkedMargins}
F^{-1}(x)-x &=& \lambda(R^{-1}(x)-x)\\
G^{-1}(y)-y &=& \frac{1}{\lambda}(S^{-1}(y)-y)
\end{eqnarray}

\end{prop}

\begin{rem}
~\\
We notice that the cumulative distribution functions sharing the same specific indetermination copula are those whose inverse function is inside the convex eligible space (coefficient respecting equation~\ref{eq:lambda}) of $F^{-1}$ or $G^{-1}$ and of the identity.
\end{rem}

\begin{proof}
~\\
See appendix~\ref{sec:proofCopuleLocale}.
\end{proof}

\begin{rem}[Transfer of the condition]
~\\
Inequality~\ref{ine_trans} which appears in the proof is quite remarkable: it expresses that having $(f,g)$ respecting hypothesis~\ref{cond:continuous} automatically means that $(r,s)$ also does, provided that $\lambda$ respects equation~\ref{eq:lambda}.
\end{rem}

In the proposition~\ref{prop:copuleLocale} we used two couples of cumulative distribution functions $((F,G),(R,S))$ representing margins and expressed a condition for the specific indetermination copula of the first couple $(F,G)$ to generate indetermination when we apply it to the second $(R,S)$; formally:
$$\mathbf{C^{+}_{F,G}}(R,S) = \mathbf{C^{+}_{R,S}}(R,S)$$
in summary: both copula agree on one point. It does not demonstrate that they are equal. 

Proposition~\ref{prop:uniqueness} completes the previous proposition~\ref{prop:copuleLocale} and ensures the copula are the same. 

\begin{prop}[Shared Indetermination Copula]
\label{prop:uniqueness}
~\\
The hypotheses of proposition~\ref{prop:copuleLocale} apply if and only if $\mathbf{C^+_{F,G}}=\mathbf{C^+_{R,S}}$.
\end{prop}
\begin{proof}
~\\
We begin with a simple remark, if for two couples $((f,g),(r,s))$ we have $\mathbf{C^+_{f,g}} = \mathbf{C^{+}_{r,s}}$ then, in particular, $\mathbf{C^{+}_{f,g}}(r,s) = \mathbf{C^{+}_{r,s}}(r,s)$, hence, proposition~\ref{prop:copuleLocale} applies and its hypotheses are satisfied. 

Now, if it exists a $\lambda$ as defined in the quoted proposition linking $F^{-1}$ to $R^{-1}$ and $G^{-1}$ to $S^{-1}$, we use it to express $\mathbf{C_{f,g}}$ as a function of $(r,s)$:
\begin{eqnarray*}
\mathbf{C^{+}_{f,g}}(u,v) &=& vF^{-1}(u) + uG^{-1}(v) -F^{-1}(u)G^{-1}(v)\\
&=& v\left[u + \lambda(R^{-1}(u)-u)\right] + uG^{-1}(v) - \left[u + \lambda(R^{-1}(u)-u)\right]G^{-1}(v)\\
&=& v\left[u + \lambda(R^{-1}(u)-u)\right] - \left[\lambda(R^{-1}(u)-u)\right]G^-1(v)\\
&=& v\left[u + \lambda(R^{-1}(u)-u)\right] - \left[\lambda(R^{-1}(u)-u)\right]\left[v+ \frac{1}{\lambda}(S^{-1}(v)-v)\right]\\
&=& uv - \left[R^{-1}(u)-u\right]\left[S^{-1}(v)-v\right]\\
&=& \mathbf{C^{+}_{r,s}}(u,v)
\end{eqnarray*}
It completes the proof. 
\end{proof}

\subsection{Maximal spread between the two copulas}
Let us estimate the difference between the two copulas we just introduced. As the indetermination copula is defined locally, we compute it for any suited couple of margins $(F,G)$ before trying to maximize it. The expression under the $\mathbb{L}^1$ norm amounts to:

\begin{eqnarray*}
\Delta_1(F,G) = ||\mathbf{C^{\times}}-\mathbf{C^{+}_{F,G}}||_1 &=&\int_{u=0}^1 \int_{v=0}^1 \left|uv - vF^{-1}(u) - uG^{-1}(v) +F^{-1}(u)G^{-1}(v) \right| \diff u \diff v\\
&=& \int_{u=0}^1 \int_{v=0}^1 \left|(F^{-1}(u)-u)(G^{-1}(v)-v)\right| \diff u \diff v\\
&=& \int_{u=0}^1 |F^{-1}(u)-u|\diff u \int_{v=0}^1 |G^{-1}(v)-v| \diff v
\end{eqnarray*}

We have a closed formula of the difference between the two couplings (the common choice of the $\mathbb{L}^1$ norm shall be motivated later). But we already know the value is null in case any of the two margins is uniform, we do not know neither the maximum, nor for which couple it is realized. Both questions are solved within Property~\ref{prop:maximalL1Diff}.

\begin{prop}
~\\
\label{prop:maximalL1Diff}
Quoting $\mathcal{E}^+$ the set of couples of cumulative distribution functions $(F,G)$ such that their densities couple $(f,g)$ respects condition~\ref{cond:continuous} we have:
$$\max_{(F,G)\in \mathcal{E}^+}{\Delta_1}(F,G) = \frac{1}{16}$$
Moreover, the couple of densities corresponding to the maximum is:
$$(f_0,g_0) = \left(u\mapsto\frac{1}{2}(1 + \delta_{u=0}),~v\mapsto\frac{1}{2}(1 + \delta_{v=0})\right)$$
\end{prop}

\begin{proof}
~\\
Working with $(F,G)\in\mathcal{E}^+$ is fairly inconvenient as it adds hypothesis on $f$ that we have to convey to $F^{-1}$ through $F$ and similarly with $G$. 
To get rid of it, we use proposition~\ref{prop:constructiveAlpha} and extract $\alpha$ together with two densities $r,s$ (respectively two cumulative distribution functions $(R,S)$). Finally, we obtain: 
$$F(u) = \alpha u + (1-\alpha) R(u)$$
as well as 
$$G(v) = (1-\alpha) v + \alpha S(v)$$

Let us report those functions in $\Delta_1(F,G)$:
\begin{eqnarray*}
\Delta_1(F,G) &=& \int_{u=0}^1 \left|F^{-1}(u)-u\right|\diff u \int_{v=0}^1 \left|G^{-1}(v)-v\right|\diff v\\
&=& \int_{u=0}^1 \left|F(u)-u\right|f(u)\diff u \int_{v=0}^1 \left|G(v)-v\right|g(v) \diff v
\end{eqnarray*}

The expression of $\Delta_1$ falsely separates $F$ and $G$: they are linked one to another using the constant $\alpha$. Though, we can maximize each integral on its own, while the respect of condition defining $\mathcal{E}^+$ will appear after. We run the computations on the "$F$-part" of $\Delta_1$:

\begin{eqnarray*}
I_f &=& \int_{u=0}^1 \left|F(u)-u\right|f(u)\diff u \\
&=& \int_{u=0}^1 \left|\alpha u + (1-\alpha)R(u)-u\right|f(u)\diff u \\
&=& (1-\alpha) \int_{u=0}^1 \left|R(u)-u\right|(\alpha + (1-\alpha)r(u))\diff u \\
&=& (1-\alpha)\left[ \alpha \int_{u=0}^1 |R(u)-u|\diff u + (1-\alpha)\int_{u=0}^1 |R^{-1}(u)-u|\diff u\right]\\
&\le & \frac{(1-\alpha)}{2}
\end{eqnarray*}

We immediately derive, conducting the same analysis on $G$:
$$\Delta_1(F,G) \le \frac{(1-\alpha)\alpha}{4} \le \frac{1}{16} \mbox{ reached for $\alpha=\frac{1}{2}$} $$

Besides, it turns out that it becomes and equality if we choose $(F,G)$ as proposed in the property, namely $f(u) = \frac{1+\delta_{u=0}}{2}$ which finishes the proof. 
\end{proof}

\begin{rem}[Transposition of the discrete case]
~\\
The couple $(f_0,g_0)$ which emphasizes most the difference between independence and indetermination appears to be the natural transposition of the discrete case. Indeed, for contingency values, each probability is $\frac{1}{2p} + \frac{1}{2} \delta_{u=0}$ where $p$ is the number of values $u$.
\end{rem}

\begin{rem}[Motivation of the $L_1$-norm]
~\\
Application of the Scheffé's Lemma enables us to convert the norm we used in the above property to the $L_{\infty}$-norm. Formally, we also have, $\forall (F,G)\in \mathcal{E}^+, \forall (U,V) \in \mathcal{B}([0,1])^2$:
$$\left| \int_{u\in U}\int_{v\in V} \mathbf{C^+_{F,G}}(u,v) \diff u \diff v - \int_{u\in U}\int_{v\in V} \mathbf{C^{\times}}(u,v) \diff u \diff v\right| \le \frac{1}{16}$$
\end{rem}

\begin{rem}[Discrete case]
~\\
In a precedent paper, we showed a similar result in the discrete case. Given two marginal laws $\mu$ and $\nu$ uniformly drawn inside the set of probability laws on $p$ and $q$ elements respectively, the expected $L^2$ norm between an indetermination and an independence coupling is less than $\frac{1}{pq}$.
\end{rem}

\subsection{Conclusions about an indetermination (or indeterminacy) copula}
Through Definition~\ref{def:CopInd}, we introduced an indetermination copula $\mathbf{C^+_{R,G}}$ specific to a given couple of margins $(F,G)$. We also showed that dependence on margins prevents us from defining a general indetermination copula: we end up with a collection of copulas parametrized by a couple of eligible margins (any respecting hypothesis~\ref{cond:continuous}). 

A priori, the association $(F,G)\mapsto \mathbf{C_{F,G}^+}$ has no reason to be injective and it is not. Precisely, a couple $(F,G)$ defines the same indetermination copula as another $(R,S)$ if and only if the four cumulative distribution functions respect the two equations~\ref{eq:linkedMargins}. It basically requires that each inverse cumulative distribution is a linear composition of the other cumulative distribution function and of the identity. On that segment of couples of cumulative distribution functions defined with an eligible $\lambda$, all the specific indetermination copulas are identical.

From a specific indetermination copula, we built up an indetermination copula that shall be defined and applied to a segment of couples of cumulative distribution functions and we also showed that it cannot be further extended. 

\section{Indetermination and average likelihood~\label{sec:likelihood}}

Given a set of $n$ realizations of a variable $W$, a common problem is to determine the underlying probability law $P_{\pi}$ that is supposed unknown and that we shall abusively quote by its corresponding density $\pi$. It often comes with a subset $\Omega\subset\Lambda$ of probability laws among which we search for the most approaching one. Here we also suppose that $P_{\pi}\in\Lambda$ applies to a product space $S$ as previously defined. 

A usual method is the Maximum Likelihood Paradigm (MLP), under which for a given realization $W=w$ and for any probability law $P_{\pi'}$ in $\Omega$, we compute the probability that $W'\sim P_{\pi'}$ belongs to $[w,w+\diff w]$. 

$$\mathbb{P}_{\pi'}(W'\in[w,w+\diff w]) = \pi'(w)\diff w$$

If we integrate on the values $W$ can take, we obtain the average likelihood $\overline{L}(\pi,\pi')$ of $W$ under $\pi'$:
\begin{mydef}[Average likelihood]
~\\
Given two probability laws $P_{\pi}$ and $P_{\pi'}$ whose densities are $\pi$ and $\pi'$ respectively we define the average likelihood between $P_{\pi}$ and $P_{\pi'}$ or abusively between $\pi$ and $\pi'$ as:
\begin{equation}
\label{eq:likelihood}
\overline{L}(P_{\pi},P_{\pi'}) = \overline{L}(\pi,\pi') = \mathbb{E}_{W\sim P_{\pi}}(\pi'(W))
\end{equation}
\end{mydef}

\begin{rem}
~\\
Finally, $\overline{L}(P_{\pi},P_{\pi'})$ rewrites $\int_{S}\pi(s)\pi'(s)\diff s$ with $S$ the set where $P_{\pi}$ and $P_{\pi'}$ take their values (non necessarily a product space). It immediately appears symmetric. 
\end{rem}

Now, motivated by the discrete properties of $\pi^+$ (see~\cite{bertrand:hal-03086553}), we link $\overline{L}$ with a discrete notion of couple matching. 

Let us go back to the discrete case, and select $\pi$ among probability laws defined on a product space of $pq$ elements with fixed margins $\mu$ and $\nu$. Then, if $W\sim\pi$, the probability of a couple matching, that is to say the case where two independent realizations $W_1=(U_1,V_1)$ and $W_2=(U_2,V_2)$ are equal, is minimal when $\pi = \pi^+ = C^+(\mu,\nu)$. 

Similarly, in the continuous case, we expect $\overline{L}(\pi,\pi)$ to be minimal when $\pi=\pi^+$ and this is precisely what happens since it amounts to computing the cost function of problem~\ref{pb:continous}:

\begin{equation}
\label{eq:costL}
\overline{L}(\pi,\pi) = \int_S \pi^2 (w)\diff w
\end{equation}

Furthermore, in the discrete case a couple matching between $\pi^+$ and a second probability law $\pi$ only depends upon the margins of $\pi$. We show this property stays true in a continuous domain in the next section.

\subsection{Indetermination prevents correlation extraction}
Going back to the definition of average likelihood given in Equation~\ref{eq:likelihood}, we suppose $\pi$ takes the form $\pi^+$ given in \ref{def:continuous+} and compute its average likelihood with any $P_h\in\Lambda$:
\begin{eqnarray}
\overline{L}(h,\pi^+) &=& \mathbb{E}_{W\sim \pi^+}(h(W))\nonumber\\
&=& \int_{x=0}^1\int_{y=0}^1 \pi^+(x,y)h(x,y)\diff x \diff y\nonumber\\
&=& \int_{x=0}^1\int_{y=0}^1 \left(f(x)+g(y)-1\right)h(x,y)\diff x \diff y\nonumber\\
&=& \int_{x=0}^1 f(x)\partial_1 P_h(x)\diff x + \int_{y=0}^1 g(y)\partial_2 P_h(y)\diff y - 1 \label{eq:likelihood+}
\end{eqnarray}
where $\partial_1 P_h$ corresponds to the density of the first margin of $P_h$ and $\partial_2 P_h$ for the second margin.

The equality of equation~\ref{eq:likelihood+} proves that the average likelihood between $h$ and $\pi^+$ only depends essentially on the four underlying margins. Furthermore leveraging on this decomposition we can derive the Theorem~\ref{th:MLE+}

\begin{theorem}
\label{th:MLE+}
~\\
Given two densities $f$ and $g$ on $S_1$, we quote $\Omega_{f,g}$ the set of densities on $S$ with margins $(f,g)$. Among $\Omega_{f,g}$, $\pi^+$ realizes the optimum of:
\begin{equation}
\min_{\pi\in\Omega_{f,g}} \overline{L}(\pi,\pi);
\end{equation}
and, for any $\pi_0$ in $\Omega_{f,g}$ we have:
\begin{equation}
\min_{\pi\in\Omega_{f,g}} \overline{L}(\pi,\pi) =  \overline{L}(\pi^+,\pi_0).
\end{equation}
\end{theorem}
\begin{proof}
~\\
The first part amounts to noticing, using Equation~\ref{eq:costL}, that it corresponds to the cost of Problem~\ref{pb:continous} which is precisely minimized among $\Omega_{f,g}$ by $\pi^+$.
For the second part, we rewrite equation~\ref{eq:likelihood+}:
$$\overline{L}(\pi^+,\pi_0) = \overline{L}(f,f)+\overline{L}(g,g)-1$$ 
so that we have:
$$\overline{L}(\pi^+,\pi_0) = \overline{L}(\pi^+,\pi^+)$$
\end{proof}

\subsection{Average likelihood as a continuous couple matching notion~\label{ssec:matchings}}
In any application, $\pi$ is not given and one must compute $\overline{L}$ using a set of $n$ realizations of $W$. We describe hereafter a method to approach $\overline{L}$; in addition we show $\overline{L}$ corresponds to the probability of couple matchings in the discrete case. 

Suppose $s_n=W_1,\ldots,W_n$ are $n$ i.i.d. points drawn under a law $P_\pi\in\Lambda$. We quote $P_n$ the empirical measure given by:
\begin{equation}
\label{eq:empiric_f}
P_n(s) = \frac{1}{n}\sum_{i=1}^n \delta_{W_i}
\end{equation}

A natural question would be, under a hypothesis of fixed margins densities $(f,g)$ and leveraging only on $P_n$, can we decide whether the underlying density $\pi$ equals $\pi^+=C^+(f,g)$ ? In the discrete case we could typically estimate the probability of a couple matching under $P_n$ and check it is minimal:
\begin{eqnarray*}
\mathbb{P}_{C^{\times}(P_n,P_n)}(W_1' = W_2') &=& \sum_{i=1}^n\frac{1}{n}P_n(W_i)\\
&=& \mathbb{E}_{W'\sim P_n}\left(P_n(W')\right)\\
&=& \overline{L}(P_n,P_n)
\end{eqnarray*}
This last equality precisely shows that the average likelihood corresponds to the discrete notion of couple matchings. Yet, in the continuous case (since $P_\pi$ has a density: $\pi$), two $W_i$ are never equal leading to a null probability for any $\pi$ whatever the correlation between margins is. This property prevents us from checking that $\overline{L}(P_n,P_n)$ is minimal since it is always null. 

We leverage on $\overline{L}(\pi^+,\pi^+) = \overline{L}(\pi^+,\pi_0)$ whatever $\pi_0\in\Omega_{f,g}$ to rely on $\overline{L}(P_n,\pi_0)$ for a fixed $\pi_0$ as a way to estimate continuous couple matchings. Later on, we will choose $\pi_0 = \pi^{\times} = C^{\times}(f,g)$.

Eventually, provided $P_n\rightarrow P_{\pi^+}$ in a certain sense, we expect $\overline{L}(P_n,\pi_0)$ to converge to $\overline{L}(\pi^+,\pi_0)$ that is to say to $\overline{L}(\pi^+,\pi^+)$ which is minimal among $\Omega_{f,g}$; section~\ref{ssec:test} formalizes the approach notably by introducing a topology on $\Lambda$.

\subsection{Statistical test based on average likelihood \label{ssec:test}}

\subsubsection{Definition of $t_n$}
Since $\pi^{\times}$ and $\pi^+$ are close in a certain sense (property~\ref{prop:maximalL1Diff}) we expect $\overline{L}(\pi^{\times},\pi^+) = l_0$ and $\overline{L}(\pi^{\times},\pi^{\times}) = l_1$ to be close. Then, using property~\ref{prop:minL2}, since $\pi^{+}$ and $\pi^{\times}$ have the same margins we know we always have $l_0\le l_1$. Eventually, provided $f$ and $g$ are not uniform,
\begin{equation}
\eta = l_1-l_0>0
\end{equation}

We shall take advantage on that difference to build up a statistical test of \textit{indeterminacy}. First, we define:
\begin{equation}
\Omega^{\times} = \left\{P_{\pi}\in\Omega_{f,g}~/~ \overline{L}(\pi,\pi^{\times}) \geq l_1\right\}
\end{equation}
and now we build up two hypotheses: $H_0$ the hypothesis $\pi=\pi^+$ while the opposite $H_1$ corresponds to $\pi\in\Omega^{\times}$. 

\begin{rem}
~\\
Using the $\rho-$topology we shall introduce shortly, we show in Proposition~\ref{prop:separable} that both hypotheses correspond to non void and separable subsets of $\Lambda$.
\end{rem}

Given $n$ observations $s_n=W_1,\ldots,W_n$, we quote $s$ the whole sequence and we define a statistical test $t_n$ depending only on the first $n$ coordinates of $s$ to distinguish between both hypotheses:

\begin{equation}
t_n(s) = \overline{L}(P_n,\pi^{\times}) = \int_S{\pi^{\times}\diff P_n(s)};
\end{equation}

where $P_n$ is defined as in Equation~\ref{eq:empiric_f}. Observing $s$, $H_0$ will be rejected if $t_n(s)$ is greater than a specific value. The just-introduced test is motivated by the combination of subsection~\ref{ssec:matchings} which allows us to interpret it as an estimation of continuous couple matchings and of Theorem~\ref{th:MLE+} which enables us (under $H_0$) to estimate $\overline{L}$ using any $\pi_0$, in particular using $\pi^{\times}$. We will study $t_n$ under the two eligible hypotheses to calibrate the threshold. Specifically, we follow the track of the second section of~\cite{groeneboom1977bahadur} dedicated to Bahadur's slope. It requires the introduction of some notations that we report hereafter.

\subsubsection{Usual notations and basic lemmas}
As common in the literature and notably in the two articles (\cite{csiszar1984sanov} and \cite{groeneboom1979large}) cited beforehand, we apply the $\tau$-topology on $\Lambda$ which is defined by the basic neighborhoods:
\begin{equation}\label{tau_neigh}
U(P,\mathcal{P},\epsilon) = \left\{Q/~\forall i,~|P(B_i)-Q(B_i)| < \epsilon\right\}
\end{equation}
where $P\in\lambda$, $\epsilon>0$ and $\mathcal{P}$ ranges over all $\mathcal{B}$-measurable partitions $\mathcal{P} = (B_1,\ldots,B_k)$ of $S$. Consequently, for any set $\Omega \subset \Lambda$, we will quote $\Omega^{\mathrm{o}}$ and $\overline{\Omega}$ the interior and the closure of $\Omega$ in the sense of the $\tau$-topology. 

A sequence of probability measures $(Q_n)_{n\in\mathbb{N}}$ converges to $Q$ for this topology if and only if $\lim_{n\rightarrow\infty} \int_{\mathbb{R}}f \diff Q_n \rightarrow \int_{\mathbb{R}}f \diff Q$ for each $\mathcal{B}$-measurable and bounded function $f:S\rightarrow\mathbb{R}$.

We also introduce the usual Kullback-Leibler divergence quoted $D_{KL}$ whose properties are gathered in \cite{van2014renyi} and which is defined by:
\begin{eqnarray}
D_{KL}(P|Q) &=&  \int_S \log\left(\frac{\diff P}{\diff Q}\right)\diff P ~~\textnormal{if }Q<<P \label{eq:KL}\\
&=& \infty ~~ \textnormal{otherwise}\nonumber
\end{eqnarray}
together with the usual conventions $\log 0 = -\infty$, $0\cdot(\pm\infty)=0$ and $\log(a/0) = \infty, ~\forall a>0$.

An usual result is the lower semi-continuous property of $D_{KL}$:
\begin{lemma}[Divergence lower semi-continuous]\label{lem:continKL}
~\\
For any $P\in\Lambda$, the function:
$$Q\mapsto D_{KL}(Q,P)$$ 
is $\tau$-lower semi-continuous.
\end{lemma}
\begin{proof}
~\\
It corresponds to Lemma~2.2 of \cite{groeneboom1979large}.
\end{proof}

Furthermore, for any subset $\Omega\subset\Lambda$ and any probability law $P\in\Lambda$ we extend the definition of the divergence:
\begin{equation}\label{eq:divergenceSpace}
D_{KL}(\Omega|P) = \inf_{Q\in\Omega} D_{KL}(Q|P)
\end{equation}
with the convention that it equals $\infty$ if $\Omega$ is empty.

Eventually, we introduce a second topology for $\Lambda$, the $\rho$-topology which is induced by the supremum metric $d$ defined on the ordered set $S=[0,1]\times[0,1]$:
\begin{equation}
d(P,Q) = \sup_{x\in S} \left|P(0,x) - Q(0,x)\right|.
\end{equation}

Enabled with those notations we list a bench of lemmas extracted from \cite{groeneboom1979large} or \cite{csiszar1984sanov}. The first lemma links the two previously defined topologies and enables us to interpret the open space of one in the second. 
\begin{lemma}[$\tau$ is finer than $\rho$]\label{lem:finer}
~\\
The $\tau$-topology is finer than the $\rho$-topology.
\end{lemma}
\begin{proof}
~\\
It corresponds to Lemma~2.1 of \cite{groeneboom1979large}.
\end{proof}

The $\rho$-topology is handier since it is associated to a distance $d$ and notably to define the subset neighborhood of a probability law $P$ as exposed below:
\begin{equation}
V_{\epsilon}(P) = \{Q/ d(P,Q)<\epsilon\}
\end{equation}

A first result regarding $t_n$ is related to the separability of $\{P_{\pi^+}\}$ and $\Omega^{\times}$. The following proposition shows the two subsets of $\Lambda$ are separable using the $\rho-$topology.

\begin{prop}\label{prop:separable}[Separable hypotheses]
~\\
Supposing $l_1>l_0$, it exists an $\epsilon >0$ such that:
\begin{equation}
\Omega^{\times}\cap V_{\epsilon}(P_{\pi^+}) = \emptyset
\end{equation}

\end{prop}
\begin{proof}
~\\
Let us suppose for a given $\epsilon>0$, $P\in V_{\epsilon}(P_{\pi^+})$ then 
\begin{eqnarray*}
\int_S \pi^{\times}\diff P &=& \int_S \pi^{\times}(\diff P -\pi^+ + \pi^+)\\
&=& l_0 + \int_S \pi^{\times}(\diff P -\pi^+)\\
& \le & l_0 + \epsilon ||\pi^{\times}||_{\infty}\\
\end{eqnarray*}
For $\epsilon$ small enough this last quantity is strictly less than $l_1$ which concludes the proof: 
$$P \in V_{\epsilon}(p_{\pi^+}) \implies P\not\in\Omega^{\times}$$
\end{proof}

With the $\rho-$topology on $\Lambda$ we can formally define the hypotheses $P_n\rightarrow P_{\pi}$ as quoted below:
\begin{equation}
P_n\xrightarrow[n\rightarrow\infty]{\rho} P_\pi.
\end{equation}
We suppose this convergence granted in the rest of the paper. Additionally, we suppose the two densities $f$ and $g$ are bounded (it implies $\pi^{\times}$ and $\pi^+$ also are).

\subsubsection{Analysis of $t_n$}
To estimate the efficiency of a test, the authors of \cite{groeneboom1977bahadur} define $G_n(t) = \mathbb{P}_0(t_n<t)$ together with its opposite $L_n(t) = 1-G_n(t)$ called the tail probability of the test based on $t_n$. Then, they apply it for $t=t(s)$ which corresponds to the current empiric observations and define the random variable:
\begin{mydef}[Tail probability of a test]
\begin{equation}
L_n(s) = L_n(t_n(s)) = 1-G_n(t_n(s)) = \mathbb{P}_0(t_n\geq t_n(s))
\end{equation}
\end{mydef}

As quoted, in their paper, "the smaller the value of $L_n(s)$ the more untenable is the hypothesis $H_0$ in the
light of the observations". Having said that they study the behavior of $L_n$ under each hypothesis ; we report the corresponding computations for our test $t_n$ hereafter. 

\paragraph{Convergence under $H_1$} Let us suppose $H_1:\pi\in\Omega^{\times}$ then we have almost surely the convergence of our statistical test. 
\begin{prop}[Convergence under $H_1$]
\label{prop:l1}
~\\
Under $H_1:\pi\in\Omega^{\times}$,
$$t_n \xrightarrow[n\rightarrow\infty]{a.s} \overline{L}(\pi^{\times},\pi) \geq l_1$$
\end{prop}
\begin{proof}
~\\
Almost surely, $P_n\xrightarrow[n\rightarrow\infty]{\rho} P_\pi \in \Omega^{\times}$.

Since $\pi^{\times}:S\rightarrow\mathbb{R}$ is a $\mathcal{B}$-measurable and bounded function it implies in particular 
$$\int_{S}\pi^{\times}\diff P_n \xrightarrow[n\rightarrow\infty]{a.s.} \int_{S}\pi^{\times}\pi = \overline{L}(\pi^{\times},\pi).$$

Eventually, given that $\pi\in\Omega^{\times}$, $\overline{L}(\pi^{\times},\pi) \geq l_1$.
\end{proof}

\begin{rem}\label{rem:sim_pi2} [Realizations under $\pi^2$]
~\\
$\pi$ being bounded on $S = [0,1]^2$, we want to simulate under $P_{\pi^2}$ whose density is $\frac{\pi^2}{\int_{S}\pi^2}$. We notice that we always have $\pi^2\le||\pi||_{\infty}\pi$. Then, applying an usual reject method we simulate a couple 
$$M=\left(\frac{\pi^2(X)}{\int_{S}\pi^2},\mathbb{U}\frac{||\pi||_{\infty}}{\int_S{\pi^2}}\pi(X)\right);$$
where $X\sim P_{\pi}$ and $\mathbb{U}$ is the uniform law on $S$.

The reject method stands that keeping $M$ only if $M_1\geq M_2$ conveys a $M_1$ under $P_{\pi^2}$: it gives a method to draw under $P_{\pi^2}$ by leveraging on a method to draw under $P_{\pi}$. 

Furthermore, in our particular application of the reject method, $M_1\geq M_2$ can be simplified:
\begin{eqnarray*}
\frac{\pi^2(X)}{\int_{S}\pi^2} &\geq & \mathbb{U}\frac{||\pi||_{\infty}}{\int_S \pi^2} \pi(X);\\
\pi(X) &\geq & \mathbb{U}||\pi||_{\infty}.
\end{eqnarray*}
Eventually it amounts to simulate independently $M=(X,U)\sim P_{\pi}\otimes\mathbb{U}$ and keep $X$ if and only if $\pi(X) \geq U ||\pi||_{\infty}$.
\end{rem}

\begin{rem}
~\\
According to the previous remark, a law under $P_{\pi^2}$ will concentrate its values around the mode of $\pi$. Indeed, the higher $\pi(X)$ is the most chance we have to keep $X$ in the reject test. 
\end{rem}
\paragraph{Divergence under $H_0$}
Let us now suppose $H_0:\pi=\pi^+$ then we similarly show:
\begin{prop}[Convergence under $H_0$]
~\\
Under $H_0:\pi=\pi^+$,
$$t_n\xrightarrow[n\rightarrow\infty]{a.s} \overline{L}(\pi^{\times},\pi^+) = l_0$$
\end{prop}

We want to estimate the asymptotic probability under $H_0$ to be close to $l_1$. A deviation proposition is given below to validate it is exponentially rare in a certain sense. 

A first result is the estimation of the probability that $t_n(s)$ with $s$ drawn under $H_0:\pi=\pi^+$ is larger than a fixed value $t$. To do so we introduce the random variables $Y = \pi^{\times}(X)$ where $X$ is under the law $\pi^+$ and its logarithmic moment generating function defined by:
\begin{equation}
\phi_0(t) = \log \mathbb{E}_0(\exp (t Y)) = \log\left(\int_S \exp(t\pi^{\times}(s)) \pi^+(s)\right)
\end{equation}
that we suppose is finite for $|t|$ small enough. Out of $\phi$ is coined its Legendre transform defined by:
\begin{equation}
I(t) = \sup_{x\in\mathbb{R}}\left\{tx-\phi(t)\right\}
\end{equation}

Eventually, we introduce the subset of $\Lambda$ in which $P_n(s)$ must stay so that we have $t_n(s)\geq t$:
\begin{equation}
\Omega_t = \left\{P\in\Lambda~s.t.~L(P,\pi^{\times}) \geq t\right\}
\end{equation}
We can now write a deviation theorem for $t_n$ under $H_0$ as stated in Theorem~\ref{th:cramer}.
\begin{theorem}[Cramér]
\label{th:cramer}
~\\
For any $t\geq l_0$, we have
\begin{equation}
\lim_{n\rightarrow\infty} \frac{1}{n}\log\left(\mathbb{P}_0(t_n\geq t)\right) = -I(t) = - D_{KL}(\Omega_t|P_{\pi^+}).
\end{equation}
\end{theorem}
\begin{proof}
~\\
We rewrite $t_n$ using $Y$:
\begin{eqnarray*}
\mathbb{P}_0 (t_n\geq t) &=& \mathbb{P}_0\left(\int_S \pi^{\times}(s)\diff P_n(s)\geq t\right)\\
&=& \mathbb{P}_0\left(\sum_{i=1}^n\pi^{\times}(X_i)\geq nt\right)\\
&=& \mathbb{P}_0\left(\sum_{i=1}^nY_i\geq nt\right)
\end{eqnarray*}
A direct application of Cramér's theorem gives us:
\begin{equation}
\lim_{n\rightarrow\infty} \frac{1}{n}\log\left(\mathbb{P}_0\left(\sum_{i=1}^nY_i\geq nt \right)\right) = -I(t) =  - D_{KL}(\Omega_t|P_{\pi^+})
\end{equation}
which concludes the proof. 
\end{proof}

The previous theorem provides an estimation of how unlikely it is to have $t_n$ higher than a fixed value $t$ under $H_0$. More interesting is to get the probability, under $H_0$, that $t_n>t_n(s)$ where $s$ is a sequence of realizations under $H_1$, namely the probability of the event "we do not reject $H_0$ while we should do" ; this estimation relies on the Bahadur slope~\cite{bahadur1967rates} as we shall precise later. 

Combining, with $s$ drawn under $P_{\pi}\in\Omega^{\times}$, $t_n(s) \xrightarrow[n\rightarrow\infty]{a.s.} \overline{L}(\pi^{\times},\pi)\geq l_1$ and theorem~\ref{th:cramer} precisely provides this estimation as stated in theorem~\ref{th:bahadur}.

\begin{theorem}[Bahadur's slope]
\label{th:bahadur}
~\\
A whole given sequence $s$ being drawn under $H_1$ through a probability $P_{\pi}\in\Omega^{\times}$ we estimate the limit of $L_n(s')$ where $s'\sim H_0:P_{\pi^+}$:
\begin{equation}
\lim_{n\rightarrow\infty} \frac{1}{n}\log\left(\mathbb{P}_0(t_n(s')\geq t_n(s))\right) = -I(l) = - D_{KL}(\Omega_l|P_{\pi^+})\geq -I(l_1) = - D_{KL}(\Omega_{l_1}|P_{\pi^+})
\end{equation}
where $l = \overline{L}(\pi^{\times},\pi)$.
\end{theorem}
\begin{proof}
~\\
$s$ being drawn under $P_{\pi}$, we have almost surely, as stated in proposition~\ref{prop:l1} 
\begin{equation}
t_n(s) \xrightarrow[n\rightarrow\infty]{a.s.} l.
\end{equation}
Furthermore, since $P_{\pi}\in\Omega^{\times}$, $l\geq l_1 > l_0$. Therefore, we select $\delta$ such that $\frac{l-l_0}{2}>\delta>0$ and the almost sure convergence provides us with an $n_0$ such that for any $n\geq n_0$,
\begin{equation}
|t_n(s)-l|\le\delta.
\end{equation}
It notably implies that for any $n\geq n_0$, 
\begin{equation}
\mathbb{P}_0(t_n(s')\geq l+\delta) \le \mathbb{P}_0(t_n(s')\geq t_n(s)) \le \mathbb{P}_0(t_n(s') \le  l-\delta).
\end{equation}
Applying theorem~\ref{th:cramer} with $t = l+\delta\geq l_0$ and $t = l-\delta\geq l_0$ we immediately obtain
\begin{equation}
-I(l+\delta) \le \lim_{n\rightarrow\infty} \frac{1}{n}\log\left(\mathbb{P}_0(t_n(s')\geq t_n(s))\right) \le  -I(l-\delta).
\end{equation}
Since $\delta$ is arbitrary it suffices to conclude after invoking the continuity of $I$. 

Finally, the lower-bound comes from the hypothesis $H_1$ itself: $P_{\pi}\in\Omega^{\times}$ requires $l\geq l_1$.
\end{proof}

The result of Theorem~\ref{th:bahadur} actually consists in showing that the Bahadur slope (see~\cite{groeneboom1977bahadur}) is $2I(l_1)$. Furthermore, property~\ref{prop:separable} shows that it is strictly positive. Eventually we have shown that the Bahadur slope of $t_n$ is given by $2I(l_1)$ and is strictly positive hence that the probability of the event "we do not reject $H_0$ while we should" is exponentially small in $n$.

\section{Conclusion}
In this article, we extended the discrete notion of coupling functions in a continuous version which operates on the densities of margins as a copula does on their cumulative distribution. Historical theoretical considerations on divergences conveyed a dual and exhaustive approach: either entropy or least squares, either independence or indetermination (called indeterminacy as well). 

The transposition of the second one in the continuous case led us to consider the existence of a classic copula to capture the associated dependence. While we coined it following natural computations based on its density we proved one cannot extract the indeterminacy dependence without a deep reference to the underlying margins. Namely, although they share the same coupling function by construction, two indeterminacy laws with two respective couples of margins will (almost) never share the same indeterminacy copula. Furthermore, we computed the exact sets of couples of margins with the same indeterminacy copula. 

Extending a discrete notion of couple matchings in a continuous space required the definition of the so-called average likelihood $\overline{L}$. Similarly to the discrete case, we demonstrated that the indeterminacy minimises the average likelihood. Using this property, we built up a test to reject or accept the indeterminacy out of the value of $\overline{L}$. Eventually, we followed a classic analysis of a statistic test to estimate its efficiency by the associated Bahadur slope and proving it is strictly positive.

\newpage 
\begin{appendix}
\section{Proof of proposition~\ref{prop:copuleLocale} \label{sec:proofCopuleLocale}}
\begin{proof}~\\
We split the proof in four parts:
\begin{itemize}
\item express the required relation between $(F,G)$ and $(R,S)$ (part~\ref{sssec:relation})
\item define two bounds of $\lambda$ according to its sign (part~\ref{sssec:bounds})
\item eliminate the case $\lambda$ negative (part~\ref{sssec:positive})
\item verify $h$ (the density  extracted from $H$) is positive (part~\ref{sssec:transfer})
\end{itemize}

\subsection{relation between $(F,G)$ and $(R,S)$ \label{sssec:relation}}
In case an indetermination occurs, the crossed derivation of $H$ equals $r+s-1$. 
Let us compute that derivative function and compare it to the required expression.

\begin{eqnarray*}
\frac{\partial^2 H(x,y)}{\partial x \partial y} &=& \frac{\partial^2 \left(R(x)*G^{-1}(S(y)) + S(y)*F^{-1}(R(x)) - F^{-1}(R(x))G^{-1}(S(y))\right)}{\partial x \partial y}\\
&=& \frac{r(x)s(y)}{g(G^{-1}(S(y)))} + \frac{r(x)s(y)}{f(F^{-1}(R(x)))} - \frac{r(x)s(y)}{f(F^{-1}(R(x)))g(G^{-1}(S(y)))}
\end{eqnarray*}

Requiring an equality with $r(x)+s(y)-1$ leads to:
\begin{align*}
r(x)+s(y)-1 &=& \frac{r(x)s(y)}{g(G^{-1}(S(y)))} + \frac{r(x)s(y)}{f(F^{-1}(R(x)))} - \frac{r(x)s(y)}{f(F^{-1}(R(x)))g(G^{-1}(S(y)))}\\
\intertext{If we divide by $r(x)s(y)$ motivated by the details we bring up in remark~\ref{req:div}:}
\frac{1}{r(x)} + \frac{1}{s(y)}-\frac{1}{r(x)s(y)} &=& \frac{1}{f(F^{-1}(R(x)))} + \frac{1}{g(G^{-1}(S(y)))} -\frac{1}{f(F^{-1}(R(x)))g(G^{-1}(S(y)))}\\
\intertext{Aggregating similar terms: }
\frac{1}{r(x)}\left(1-\frac{1}{s(y)}\right) &=& \frac{1}{f(F^{-1}(R(x)))}\left(1-\frac{1}{g(G^{-1}(S(y)))}\right) + \frac{1}{g(G^{-1}(S(y)))} - \frac{1}{s(y)}
\end{align*}

To simplify computations, we introduce some notations: $\alpha$ represents the "simple" part and $\beta$ its "complex" analogue; the $x$ or $y$ subscript is natural.
\begin{eqnarray*}
\alpha_x = 1-\frac{1}{r(x)} && \alpha_y = 1-\frac{1}{s(y)}\\
\beta_x = 1-\frac{1}{f(F^{-1}(R(x))} && \beta_y = 1-\frac{1}{g(G^{-1}(S(y))}
\end{eqnarray*}

It rewrites:
\begin{eqnarray*}
\frac{1}{r(x)}\alpha_y &=& \frac{1}{f(F^{-1}(R(x)))}\beta_y + \alpha_y-\beta_y\\
\left(\frac{1}{r(x)}-1\right)\alpha_y &=& \left(\frac{1}{f(F^{-1}(R(x)))}-1\right)\beta_y\\
\alpha_x\alpha_y &=& \beta_x\beta_y
\end{eqnarray*}

If we divide by $\alpha_x\beta_y$ on each side to isolate terms according to their underlying variable (please refer to remark~\ref{req:div2} for a justification of the division) we deduce the existence of a constant $\lambda$ such that:
$$\frac{\alpha_y}{\beta_y} = \frac{\beta_x}{\alpha_x} = \lambda$$

As requested by the expected conclusion, we have to go back to the cumulative distribution functions. They are contained in $\alpha$ and $\beta$ and the relation $\frac{\beta_x}{\alpha_x} = \lambda$ can be written:
\begin{equation*}
r(x) -\frac{r(x)}{f(F^{-1}(R(x))} = \lambda(r(x)-1)
\end{equation*}

If we integrate on $[0,x]$:
\begin{equation*}
R(x) -F^{-1}(R(x)) = \lambda(R(x)-x)
\end{equation*}

As $R$ is a cumulative distribution function whose derivative (its density) never equals $0$ (hypothesis already mentioned), $R$ has an inverse function so that for all $x$ :
\begin{equation}\label{R_x}
F^{-1}(x) = x(1-\lambda) + \lambda R^{-1}(x)
\end{equation}

and symmetrically for any $y$:
\begin{equation}\label{S_y}
G^{-1}(y) = y(1-\frac{1}{\lambda}) + \frac{1}{\lambda}S^{-1}(y)
\end{equation}

Equations~\ref{R_x} et \ref{S_y} can be written:
\begin{eqnarray}\label{RS_transfert}
F^{-1}(x)-x &=& \lambda(R^{-1}(x)-x)\\
G^{-1}(y)-y &=& \frac{1}{\lambda}(S^{-1}(y)-y)\nonumber
\end{eqnarray}

We eventually check by applying $C^+_{f,g}$ to $(R,S)$ that those relations are sufficient to formally have:
$$H = C^+_{f,g}(R,S)(x,y) = xS(y) +yR(x) - xy = C^+_{r,s}(R,S)(x,y)$$
This does not ensures $H$ is an eligible cumulative distribution function as we shall precise it below. 

\subsection{Two bounds of $\lambda$ \label{sssec:bounds}}
We defined a $\lambda$ linking the four cumulative distribution functions as announced in the proposition and verified it suffices to formally have the expected equality. Yet, defining $(R,S)$ from $(F,G)$ using equation~\ref{R_x} and equation~\ref{S_y} does not always generate appropriate functions. 

Indeed, the non-decreasing hypothesis on $R$ as well as on $S$ conveys hypothesis on $\lambda$. 
\bigbreak
Let us first suppose that $\lambda$ is positive then $\lambda R^{-1}(x) = F^{-1}(x)-x(1-\lambda)$ is a non-decreasing function as well so that for all $x$, $\frac{1}{f(x)}\ge 1-\lambda$. Applying similar computations on $y$ for $g$ we eventually obtain:
\begin{equation}\label{C_pos}
\frac{max(g)}{max(g)-1}\geq \lambda \geq 1 - \frac{1}{max(f)}
\end{equation}
\bigbreak
Similarly, if we suppose $\lambda$ is negative, we obtain: 
\begin{equation}\label{C_neg}
1 - \frac{1}{min(f)} \geq \lambda \geq \frac{min(g)}{min(g)-1}
\end{equation}
Those values exist unless $g$ equals $1$ in which case $G=Id$ and $C^+_{f,g} = C^{\times}_{f,g}$. Hence, without restriction, $min(f) < 1 < max(f), min(g) < 1 < max(g)$.

\subsection{$\lambda$ positive \label{sssec:positive}}
Eventually, we notice that the hypothesis~\ref{cond:continuous} which guarantees that we can couple $r,s$ using indetermination is not automatically transferred a priori. Then, as equations~\ref{RS_transfert} are continuous around $\lambda=1$, the hypothesis on $(f,g)$ should be enough as soon as $\lambda$ is close to $1$. Let us unfold this remark. 

We go back to equations~\ref{R_x} and \ref{S_y} in order to extract an expression of the densities $r$ and $s$.
\begin{eqnarray*}
\frac{1}{f(F^{-1}(x))} &=& (1-\lambda) + \frac{\lambda}{r(R^{-1}(x))}\\
r(R^{-1}(x)) &=& \frac{1}{\frac{1}{\lambda f(F^{-1}(x))}+\frac{1}{\lambda}-1}\\
r(\tilde{x}) &=& \frac{1}{\frac{1}{\lambda f(F^{-1}(R(x)))}+\frac{1}{\lambda}-1}
\end{eqnarray*}

where $\tilde{x}=R^{-1}(x)$ runs over all the segment $[0,1]$ as well as $x$. We similarly have:
\begin{equation*}
s(\tilde{y}) = \frac{1}{\frac{\lambda}{g(G^{-1}(S(y)))}+\lambda-1}
\end{equation*}

Hypothesis~\ref{cond:continuous} requires that for any $(x,y)$
\begin{align*}
r(x)+s(y)-1 &\ge  0&\\
\intertext{equivalently: }
r(\tilde{x})+s(\tilde{y})-1 &\ge 0&\\
\intertext{it precisely means: }
\frac{1}{\frac{1}{\lambda f(x)}+\frac{1}{\lambda}-1} + \frac{1}{\frac{\lambda}{g(y)}+\lambda-1} &\ge 1&\\
\end{align*}
 Using inequalities~\ref{C_neg} and \ref{C_pos} we know that $r,s \ge 0$, so, we can multiply each side by $s$ leading to
\begin{equation*}
\frac{\frac{\lambda}{g(y)}+\lambda-1}{\frac{1}{\lambda f(x)}+\frac{1}{\lambda}-1} + 1 \ge \frac{\lambda}{g(y)}+\lambda-1.
\end{equation*}
Multiplying one more time on each side but by $r$:
\begin{equation*}
\frac{\lambda}{g(y)}+\lambda-1 + \frac{1}{\lambda f(x)}+\frac{1}{\lambda}-1 \ge 
 \left(\frac{\lambda}{g(y)}+\lambda-1\right)\left(\frac{1}{\lambda f(x)}+\frac{1}{\lambda}-1\right)
\end{equation*}\\
Which can be finally rewritten:
\begin{eqnarray*}
\left(\frac{1}{\lambda f}+\frac{1}{\lambda}-2\right)\left(\frac{\lambda}{g}+\lambda-2\right) &\le & 1\\
\left(\frac{1}{f}+1-2 \lambda \right) \left(\frac{1}{g}+1-\frac{2}{\lambda}\right) &\le & 1
\end{eqnarray*}
The last equality definitely excludes $\lambda\le 0$ leading to the bounds in equation~\ref{C_pos}.

\subsection{Transfer of hypothesis~\ref{cond:continuous} \label{sssec:transfer}}
At the end of the previous section, we concluded that $\lambda$ is positive. Yet, the computations began by requiring $r(x)+s(y)\ge 0$; let us resume them to obtain a second degree equation to solve. Indeed, the last line is equivalent to:
\begin{equation*}
\left(\frac{1}{min(f)}+1-2 \lambda\right) \left(\frac{1}{min(g)}+1-\frac{2}{\lambda}\right) \le 1
\end{equation*}
Which can be written (with obvious notations):
\begin{equation*}
(2 m_f+2 m_g)\lambda^2 -(1+m_f+m_g+4 m_f m_g)\lambda +2m_f+2m_g \ge 0
\end{equation*}
and that last inequality is verified if and only if:
\begin{equation*}
\frac{1+m_f+m_g+4 m_f m_g}{2m_f+2m_g} \le 2
\end{equation*}
Additionally we notice:
\begin{equation}\label{ine_trans}
\frac{1+m_f+m_g+4 m_f m_g}{2m_f+2m_g} \le \frac{1+3m_f+3m_g}{2m_f+2m_g} = \frac{3}{2}+\frac{1}{2m_f+2m_g} \le 2
\end{equation}
The last inequality coming from the condition~\ref{cond:continuous} valid on $(f,g)$ which is thus transferred to $(r,s)$ as soon as $\lambda$ is suitable (meaning verifying equation~\ref{eq:lambda}).

\begin{rem}[Division by $rs$]\label{req:div}
~\\
If $r$ equals $0$, hypothesis~\ref{cond:continuous} rewrites $r\ge 1$ hence $R=Id$. In that case, $C_{r,s}^+ = C^{\times}$. It can be shared with $C_{f,g}^+$ if and only if $F$ or $G$ equals $Id$ but it doesn't cover any interesting case.
\end{rem}

\begin{rem}[Division by $1-\frac{1}{f}$]\label{req:div2}
~\\
As we supposed $f$ and $g$ different from $1$, we can set ourselves on a interval where they never equal $1$. The division is then allowed and the proof similarly conveys the expected constant $\lambda$.
\end{rem}
\end{proof}
\end{appendix}	

\newpage
\bibliographystyle{acm}
\bibliography{Biblio}

\begin{thebibliography}{10}

\bibitem{AHP09a}
{\sc Ah-Pine, J.}
\newblock On aggregating binary relations using 0-1 integer linear programming.
\newblock In {\em ISAIM\/} (2010), pp.~1--10.

\bibitem{bahadur1967rates}
{\sc Bahadur, R.~R.}
\newblock Rates of convergence of estimates and test statistics.
\newblock {\em The Annals of Mathematical Statistics 38}, 2 (1967), 303--324.

\bibitem{bertrand2020statistical}
{\sc Bertrand, P., Broniatowski, M., and Marcotorchino, J.-F.}
\newblock {Independence versus Indetermination: basis of two canonical
  clustering criteria}.
\newblock working paper or preprint, July 2020.

\bibitem{bertrand:hal-03086553}
{\sc Bertrand, P., Broniatowski, M., and Marcotorchino, J.-F.}
\newblock {Logical indetermination coupling:a method to minimize drawing
  matches and its applications}.
\newblock working paper or preprint, Dec. 2020.

\bibitem{csiszar1984sanov}
{\sc Csisz{\'a}r, I.}
\newblock Sanov property, generalized {I}-projection and a conditional limit
  theorem.
\newblock {\em The Annals of Probability\/} (1984), 768--793.

\bibitem{csiszar1991least}
{\sc Csisz{\'a}r, I., et~al.}
\newblock Why least squares and maximum entropy? an axiomatic approach to
  inference for linear inverse problems.
\newblock {\em The annals of statistics 19}, 4 (1991), 2032--2066.

\bibitem{groeneboom1977bahadur}
{\sc Groeneboom, P., and Oosterhoff, J.}
\newblock Bahadur efficiency and probabilities of large deviations.
\newblock {\em Statistica {N}eerlandica 31}, 1 (1977), 1--24.

\bibitem{groeneboom1979large}
{\sc Groeneboom, P., Oosterhoff, J., and Ruymgaart, F.~H.}
\newblock Large deviation theorems for empirical probability measures.
\newblock {\em The Annals of Probability\/} (1979), 553--586.

\bibitem{HUY15}
{\sc Huyot, B., Mabiala, Y., and Marcotorchino, J.-F.}
\newblock Optimal transport, independence versus indetermination duality,
  impact on a new copula design.
\newblock {\em International Conference on Geometric Science of Information\/}
  (2015), 68--76.

\bibitem{MAR84}
{\sc Marcotorchino, J.-F.}
\newblock Utilisation des comparaisons par paires en statistique des
  contingences.
\newblock {\em Publication du Centre Scientifique IBM de Paris et Cahiers du
  Séminaire Analyse des Données et Processus Stochastiques Université Libre
  de Bruxelles\/} (1984), 1--57.

\bibitem{MAR86}
{\sc Marcotorchino, J.-F.}
\newblock Maximal association theory as a tool of research.
\newblock {\em Classification as a tool of research , W.Gaul and M. Schader
  editors, North Holland Amsterdam\/} (1986).

\bibitem{MarcoGSIDual}
{\sc Marcotorchino, J.-F., and Conde-C{\'e}spedes, P.}
\newblock Optimal transport and minimal trade problem, impacts on relational
  metrics and applications to large graphs and networks modularity.
\newblock In {\em International Conference on Geometric Science of
  Information\/} (2013), Springer, pp.~169--179.

\bibitem{MAM79}
{\sc Marcotorchino, J.-F., and Michaud, P.}
\newblock {\em Optimisation en Analyse Ordinale des Données}.
\newblock Ed Masson, Paris, 1979.

\bibitem{Sklar73}
{\sc Sklar, A.}
\newblock Random variables, joint distribution functions, and copulas.
\newblock {\em Kybernetika 9}, 6 (1973), 449--460.

\bibitem{van2014renyi}
{\sc Van~Erven, T., and Harremos, P.}
\newblock R{\'e}nyi divergence and kullback-leibler divergence.
\newblock {\em IEEE Transactions on Information Theory 60}, 7 (2014),
  3797--3820.

\end{thebibliography}

\end{document}